\documentclass[preprint,12pt]{elsarticle}

\usepackage{amssymb}

\usepackage{graphicx}          

\usepackage{bm}
\usepackage{amsmath,amssymb}
\usepackage{mathrsfs}
\usepackage{amsfonts}
\usepackage{color}

\journal{Journal of the Franklin Institute}

\begin{document}

\begin{frontmatter}



\title{Closed-Loop Designed Open-Loop Control of Quantum Systems: An Error Analysis}


\author[inst1]{Shikun Zhang}
\author[inst1,inst2]{Guofeng Zhang}
\affiliation[inst1]{organization={Department of Applied Mathematics, The Hong Kong Polytechnic University},
            addressline={Hung Hom, Kowloon, Hong Kong Special Administrative Region of China}, 
            country={China}}

\affiliation[inst2]{organization={Shenzhen Research Institute, The Hong Kong Polytechnic University},
            city={Shenzhen},
            postcode={518057}, 
            country={China}}

\begin{abstract}
Quantum Lyapunov control, an important class of quantum control methods, aims at generating converging dynamics guided by Lyapunov-based theoretical tools. However, unlike the case of classical systems, disturbance caused by quantum measurement hinders direct and exact realization of the theoretical feedback dynamics designed with Lyapunov theory. Regarding this issue, the idea of closed-loop designed open-loop control has been mentioned in literature, which means to design the closed-loop dynamics theoretically, simulate the closed-loop system, generate control pulses based on simulation and apply them to the real plant in an open-loop fashion. Based on bilinear quantum control model, we analyze in this article the error, i.e., difference between the theoretical and real systems' time-evolved states, incurred by the procedures of closed-loop designed open-loop control. It is proved that the error at an arbitrary time converges to a unitary transformation of initialization error as the number of simulation grids between 0 and that time tends to infinity. Moreover, it is found that once the simulation accuracy reaches a certain level, adopting more accurate (thus possibly more expensive) numerical simulation methods does not efficiently improve convergence. We also present an upper bound on the error norm and an example to illustrate our results. 
\end{abstract}



\begin{keyword}
quantum Lyapunov control \sep closed-loop design \sep open-loop control \sep error analysis
\end{keyword}

\end{frontmatter}


\section{Introduction}
There has been intensive research on applying Lyapunov's wisdom to quantum control. Over the past two decades, Lyapunov-based quantum control has been studied in the context of manipulating quantum states \citep{KUANG200898,CHEN2017439,CONG20209220,MIRRAHIMI20051987,BEAUCHARD2007388,PhysRevA.96.033803,PhysRevA.102.022603,8481533}; and quantum gates \citep{doi:10.1080/00207179.2016.1161830,HOU2014699,NOURALLAH2022}. Moreover, Lyapunov-related theories have been applied to investigate the stability of quantum systems \citep{5437308,doi:10.1063/1.4884300,6697323,8264475,EMZIR2022110660}, and also quantum filtering \citep{GE20121031,doi:10.1137/100801287}.

Although Lyapunov methods are often linked with the notion of feedback regarding classical systems, ``quantum Lyapunov control", however, often refers to open-loop control strategies in literature. We shall give an explanation herein. Theory of quantum physics dictates that the dynamics of quantum systems be linear in the Schr\"odinger picture. Since linear quantum dynamics is limited in what it can achieve, it is hoped that nonlinearity may be introduced via feedback protocols. Similar to the classical scenario, when we have at hand a system of differential equations that describe a quantum system (i.e., deterministic master equations), we may design feedback protocols which fulfill a certain goal guided by Lyapunov theory, resulting in a theoretical, closed-loop and possibly nonlinear dynamical system (also a deterministic master equation). However, different from the classical scenario, this theoretical system may \textit{not} be possible to realize in practice. This is because feedback protocols, if viewed as a function of the system's real-time state, requires the acquisition of the system's state in real-time when being implemented. However, quantum mechanics says that measurement causes non-negligible disturbance to the system, thus altering the deterministic master equation theoretically designed.

Regarding this issue, a practical solution has been mentioned in literature, e.g., \citep{DONG2022,1272601,KUANG2017164,KUANG2021109327,ZHAO20121833,PhysRevA.86.022321}. The idea is to first simulate the closed-loop system in a computing device, then generate control pulses based on simulation results, and finally apply the pre-calculated control pulses to the real plant in an open-loop fashion. It is now clear how quantum Lyapunov control relates to open-loop stategy and what we mean by ``\textit{closed-loop designed open-loop control}" in the title: the control pulses are designed based on theoretical closed-loop models but actuated in real open-loop plants.

At this point, we make a detour to give a few remarks on quantum Lyapunov control. This detour, despite being not too closely related to the theoretical results in this article, provides a proof of significance for our work. We ask the following question: what role does this kind of open-loop control play? In terms of control theory and practice, the feedback mechanism seems to be a more reliable option compared with open-loop control. However, the special feature of quantum systems decides that implementing feedback is not free even for systems with modest scales. It remains an experimental challenge to apply effective feedback to large-scale quantum systems, e.g., arrays comprising of thousands of qubits. Therefore, in the near future, open-loop control is expected to remain as a practical compromise. Since we are making this compromise, it seems only natural to develop more techniques to enhance the utility of open-loop quantum control. Quantum Lyapunov control intersects with such an endeavor. Supported by Lyapunov theory, open-loop control pulses that aim at approximating converging dynamics may be generated, which may find applications in tasks including quantum state preparation.

The closed-loop design stage of quantum Lyapunov control has been intensively studied, e.g., \citep{KUANG200898,CHEN2017439,CONG20209220,MIRRAHIMI20051987,doi:10.1080/00207179.2016.1161830}. However, to the best of our knowledge, the above mentioned ``simulate-generate-apply" procedure still needs investigation with mathematical rigor. A major question naturally arises: since the real plant is, in fact, under open-loop control, then how does the real time-evoluted states differ from the ones associated with the theoretical feedback system? This is exactly what we concern here.

In this article, we analyze the error between the closed-loop designed theoretical feedback system and the real system driven by open-loop control input generated from simulation results. We focus on the following three factors: 1) the difference between the real and theoretical initial state; 2) the accuracy of the numerical approximation methods applied to simulate the theoretical feedback system; 3) the number of simulation grids, and study how they affect the error under concern. Our work is expected to provide useful guidance to the implementation of closed-loop designed open-loop quantum control, since it is essential to know how reliable a control method is prior to its application.

The rest of this article is organized as follows. Section 2 introduces some preliminaries and formulates the error analysis problem. Section 3 contains three theorems on error analysis. Section 4 presents an example and Section 5 concludes this article.

\section{Preliminaries and Problem Formulation}
In this section, we introduce some preliminaries and present the mathematical formulation of the error analysis problem. 

Consider a Hilbert space $\mathcal{H} \triangleq \mathbb{C}^{N_S}$, $N_S \geq 2$. Let $\mathcal{B}(\mathcal{H})$ be the set of linear operators on $\mathcal{H}$, and $\mathcal{D}(\mathcal{H}) \subset \mathcal{B}(\mathcal{H})$ be the set of positive semi-definite, trace-one linear operators on $\mathcal{H}$. An element $\rho \in \mathcal{D}(\mathcal{H})$ is called a ``density operator" or a ``density matrix", which represents a quantum state. \textcolor{blue}{In this article, i denotes the imaginary unit. $[\cdot,\cdot]$ stands for the commutator between two linear operators. That is, for linear operators $A$ and $B$, $[A,B] \triangleq AB-BA$.}

The dynamics of natural closed quantum systems can be described by the Liouville equation, which is written as:
\begin{equation} \label{lineareq}
   \dot{\rho}= -\text{i}[H,\rho],
\end{equation}
where $H \in \mathcal{B}(\mathcal{H})$ is a Hermitian operator called the ``Hamiltonian". Its time-varying form is written as:
\begin{equation} \label{lineareqtvna}
   \dot{\rho}= -\text{i}[H(t),\rho].
\end{equation}
By ``natural", we refer to quantum systems which evolve without human intervention. The Hamiltonians which describe the internal energy of such systems are specified by Nature. The dynamics induced by (\ref{lineareq}) is linear. Also, any evolution trajectory $\rho (t)$ governed by (\ref{lineareq}) which starts from an element in $\mathcal{D}(\mathcal{H})$ remains in $\mathcal{D}(\mathcal{H})$. 

Instead of having $H$ specified by Nature, we may control the evolution of a quantum system by designing a specific Hamiltonian $H_{\text{spec}}$. We may even modulate system Hamiltonians so that they vary with time, resulting in equations with the following form:
\begin{equation} \label{lineareqtv}
   \dot{\rho}= -\text{i}[H_{\text{spec}}(t),\rho].
\end{equation}
A specific choice of $H_{\text{spec}}(t)$ is expressed as:
\begin{equation} \label{openloop}
    H_{\text{spec}}(t)=H_0 + \sum_{k=1}^M v_k (t) H_k ,
\end{equation}
where the Hermitian operator $H_0 \in \mathcal{B}(\mathcal{H})$ is called the ``drift Hamiltonian" and Hermitian operators $H_1,...,H_M \in \mathcal{B}(\mathcal{H})$ are called ``control Hamiltonians". The real functions $v_k (\cdot)$, $1 \leq k \leq M$, are control inputs to be designed. Eq. (\ref{lineareqtv}) with time-varying Hamiltonian (\ref{openloop}) is referred to as the \textit{bilinear control model of closed quantum systems}.  

However, the dynamics of eq. (\ref{lineareqtv}) is still linear. Also, from the perspective of Systems and Control, eq. (\ref{lineareqtv}) depicts an open-loop system without feedback. Since it is acknowledged that the linear dynamics of open-loop systems may be restricted in functionality, we may wish to design feedback systems of the following form:
\begin{equation}\label{feedback}
    \dot{\rho}= -\text{i}[H_0 + \sum_{k=1}^M u_k (\rho) H_k,\rho],\quad \rho (0)=\rho_0,
\end{equation}
where function(s) $u_k (\cdot)$, $1 \leq k \leq M$, is (are) from $\mathcal{B}(\mathcal{H})$ to $\mathbb{R}$. It is clear from (\ref{feedback}) that, by designing feedback protocol(s) $u_k (\cdot)$, $1 \leq k \leq M$, more versatile dynamical trajectories (e.g., nonlinear dynamics) can be drawn in theory. By employing Lyapunov techniques, we may design sophisticated feedback protocols that result in evolution trajectories which we desire, for example, trajectories that converge to a target quantum state in the long time limit or finite time. 

However, quantum mechanics brings bad news: chances are high that the dynamics we design \textit{in theory} can only exist \textbf{in theory}. By examining (\ref{feedback}), it is seen that real-time information about the system's state needs to be fed into feedback protocols, which further indicates that real-time measurement of the system is essential. However, quantum mechanics says that measuring a quantum system may cause non-neligible disturbance to it, which in turn indicates that system (\ref{feedback}) may not be realized in practice if feedback protocols are nontrivial (i.e., not constant functions). 

In terms of this issue, the idea of ``closed-loop designed open-loop control" has been mentioned in literature as a practical solution, e.g., see \citep{DONG2022,1272601,KUANG2017164,KUANG2021109327,ZHAO20121833,PhysRevA.86.022321}. We mathematically formulate this idea, which paves the way for subsequent error analysis.

We believe that theoretically designed control protocols are meant to be applied in the real world, and all real-world protocols are executed over a finite-time interval. Let $T>0$ be a finite real number, and $[0,T]$ be the time interval during which control actions on the plant system are to be exerted. The procedure goes as follows.

\textbf{Step 1.} System (\ref{feedback}) is simulated in a digital device with a predetermined numerical approximation method. The interval $[0,T]$ is divided into $N$ intervals, i.e., $[0,T]=\bigcup_{n=0}^{N-2}[\frac{nT}{N},\frac{(n+1)T}{N})\bigcup [\frac{(N-1)T}{N},T]$. Let us denote $t_n \triangleq \frac{nT}{N}$. The numerical method (e.g., Euler method, fourth-order Runge-Kutta method, or a linear multistep method) generates a set of discrete data points $\{\theta_j\}_{j=0}^N$, with $\theta_0=\rho_0$ and $\theta_j$ approximating the solution of (\ref{feedback}) at time $t_j$, $1\leq j\leq N$. (For example, Euler method works in the following way: $\theta_{n+1}=\theta_n - \text{i}\frac{T}{N}[H_0+\sum_{k=1}^M u_k(\theta_n)H_k,\theta_n]$, $0 \leq n \leq N-1$.)

\textbf{Step 2.} Based on $\{\theta_j\}_{j=0}^N$, $M$ control function(s) $\tilde{u}_1 (\cdot),...,\tilde{u}_M (\cdot)$ on $[0,T]$ are generated, with $\tilde{u}_j (t)=u_j (\theta_n)$, $t \in [t_n, t_{n+1})$, $0 \leq n \leq N-1$, and $1 \leq j \leq M$. 

\textbf{Step 3.} The control functions are applied to the real-world plant system with initial state $\sigma_0$ in an open-loop fashion. Let $\sigma_N (t)$ be the state trajectory of the plant system. It holds that $\sigma_N (0)=\sigma_0$ and 
\begin{equation}\label{realsystem}
    \dot{\sigma_N}=-\text{i}[H_0+\sum_{k=1}^M \tilde{u}_k H_k,\sigma_N],
\end{equation}
on each interval $(t_n,t_{n+1})$, $0\leq n \leq N-1$.

As can be seen, system (\ref{feedback})(\textit{theory}) is nonlinear if feedback protocol(s) is (are) nontrivial, while system (\ref{realsystem}) (\textit{practice}) must be linear (though possibly time-varying) since the control input(s) is (are) predetermined. It is clear that the dynamics of the two systems cannot always agree with each other. Let us denote:
\begin{equation}\label{Error}
    e_N (t)=\rho(t;\rho_0,0)-\sigma_N (t),\quad t \geq 0,
\end{equation}
which represents the error we will be analysing in the next section. In (\ref{Error}), $\rho(t;\rho_0,0)$ denotes the solution of (\ref{feedback}). From Steps 1-3, the following three factors are extracted: 1) $\rho_0-\sigma_0$, which is the initial error; 2) accuracy of the numerical approximation method applied; and 3) $N$, which shows how refined the partition of the entire simulation time slot is (and thus how heavy we invest into computation). We are interested in analysing how these three factors affect the error $e_N$. 

We now close this section by a brief summary. To achieve a certain goal, feedback protocol(s) is (are) designed, resulting in a desired \textbf{theoretical} evolution trajectory governed by (\ref{feedback}). As much as we hope to see that trajectory in practice, quantum mechanics has set its bars, and we thus resort to ``closed-loop designed open-loop control" specified by Steps 1-3. The error $e_N$ incurred by these steps is what we concern in this article.

\section{Error Analysis}
The error $e_N$ is analyzed in this section. Before presenting the main results, some preparation works are in order. 

It is first assumed that $u_k(\cdot)$ is a $C^2$-differentiable mapping from $\mathcal{B}(\mathcal{H})$ to $\mathbb{R}$, for $1 \leq k \leq M$. Then, let us denote $g_k^{\rho_0}(t)\triangleq u_k\big(\rho(t;\rho_0,0)\big)$, $t\geq 0$, $1 \leq k \leq M$. Consider the following dynamical system:
\begin{equation}\label{linearlocal}
    \dot{\omega}=-\text{i}[H_{\rho_0}(t),\omega],
\end{equation}
where
\begin{equation}\label{tvhamiltonian}
    H_{\rho_0}(t)\triangleq H_0+\sum_{k=1}^M g_k^{\rho_0}(t)H_k,\quad t \geq 0.
\end{equation}
We define the superoperator $\mathcal{U}_{\rho_0}[t,\tau]$ with two real operands as the state transition superoperator of system (\ref{linearlocal}). Moreover, the following notation is made:
\begin{equation}\label{Herror}
    E(N,n)\triangleq \sum_{k=1}^{M}\Bigg(g_k^{\rho_0}\bigg(\frac{nT}{N}\bigg)-\tilde{u}_k \bigg(\frac{nT}{N}\bigg)\Bigg)H_k.
\end{equation}
\textcolor{blue}{It can be verified that $E(N,n)=H_{\rho_0}(\frac{nT}{N})-\big(H_0+\sum_{k=1}^M \tilde{u}_k (\frac{nT}{N})H_k\big)$. On one hand, $H_{\rho_0}(\cdot)$ is the theoretically designed Hamiltonian function. On the other hand, $H_0+\sum_{k=1}^M \tilde{u}_k (\cdot)H_k$ is the Hamiltonian function of the real system which depends on $\tilde{u}_k$: the control input determined by simulation results (Step 2). $E(N,n)$ thus stands for the difference between two Hamiltonian functions at time $n \frac{T}{N}$. Therefore, $\Vert E(N,n) \Vert$ reflects the accuracy of simulation, which may be improved by consuming more computational resources.}

Finally, $\forall A \in \mathcal{B}(\mathcal{H})$, we denote
\begin{equation}\label{norm}
    \Vert A \Vert \triangleq \sqrt{\text{tr}(A^\dagger A)}.
\end{equation}

\subsection{Error Asymptotic Analysis as $N$ Tends to Infinity}
We are now in the position to present the first result.

\newtheorem{thm}{Theorem}
\begin{thm}
$\forall T>0$ fixed, if $\exists M_1 >0$, such that \textcolor{blue}{$\Vert E(N,n)\Vert \leq \frac{M_1 T}{N}$} holds $\forall N \in \mathbb{N}^+$ and $0\leq n \leq N-1$, then
\begin{equation}\label{s1}
    \lim_{N \to +\infty} e_N (T)=\mathcal{U}_{\rho_0}[T,0](\rho_0-\sigma_0).
\end{equation}
\end{thm}
\newproof{proof}{Proof}
\begin{proof}
Let us denote 
\begin{equation}
    h\triangleq \frac{T}{N}.
\end{equation}
Now consider the solution $\rho(\rho_0;t,0)$ of system (\ref{feedback}). This solution satisfies:
\begin{equation}\label{solutionequal}
    \frac{d}{dt}\rho(\rho_0;t,0)=-\text{i}[H_0+\sum_{k=1}^M g_k^{\rho_0}(t)H_k,\rho(\rho_0;t,0)].
\end{equation}
Eq. (\ref{solutionequal}) indicates that $\rho(\rho_0;t,0)$ is also a solution of (\ref{linearlocal}) that passes through $(\rho_0,0)$. Therefore, for $t\in \big[nh,(n+1)h\big)$, $0\leq n \leq N-1$, we have
\begin{equation}\label{addminus}
    \begin{aligned}
    \frac{d}{dt}\big(\rho(\rho_0;t,0)\!-\!\sigma_N (t)\big)&=-\text{i}[H_{\rho_0}(t),\rho(\rho_0;t,0)]+\text{i}[H_N^n,\sigma_N]\\
    &=-\text{i}[H_{\rho_0}\!(t),\rho(\rho_0;t,0)-\sigma_N]-\text{i}[H_{\rho_0}\!(t)-H_N^n,\sigma_N],
    \end{aligned}
\end{equation}
where $H_N^n \triangleq H_0+\sum_{k=1}^M \tilde{u}_k (nh)H_k$, $0 \leq n \leq N-1$. As a result, for $0\leq n \leq N-1$,
\begin{multline}\label{induction1}  
    e_N \!((n\!+\!1)h)\!=\mathcal{U}_{\rho_0}[(n+1)h,nh](e_N(nh))\\
    +\!\int_{nh}^{(n+1)h}\! \mathcal{U}_{\rho_0}[(n\!+\!1)h,\!s](-\text{i}[H_{\rho_0}\!(s)\!-\!H_N^n,\!\sigma_N\!(s)])ds.
\end{multline}
Next, we prove by mathematical induction that 
\begin{equation}\label{induction2}
    e_N \!(T)\!=\mathcal{U}_{\rho_0}[T,0](\rho_0-\sigma_0)+\!\int_{0}^{T}\! \mathcal{U}_{\rho_0}[T,s](-\text{i}[H_{\rho_0}\!(s)\!-\!H_N (s),\!\sigma_N\!(s)])ds,   
\end{equation}
where $H_N (s)\triangleq H_0+\sum_{k=1}^M \tilde{u}_k (s)H_k$, $s \in [0,T)$.
From (\ref{induction1}), it is clear that
\begin{equation}\label{induction3}   
    e_N \!(h)=\mathcal{U}_{\rho_0}[h,0](\rho_0-\sigma_0)+\!\int_{0}^{h}\! \mathcal{U}_{\rho_0}[h,\!s](-\text{i}[H_{\rho_0}\!(s)-H_N (s),\!\sigma_N\!(s)])ds.
\end{equation}
Suppose that, for $1 \leq J \leq N-1$, it holds that 
\begin{equation}\label{induction4}
    e_N \!(Jh)=\mathcal{U}_{\rho_0}[Jh,0](\rho_0-\sigma_0)+\!\int_{0}^{Jh}\! \mathcal{U}_{\rho_0}[Jh,\!s](-\text{i}[H_{\rho_0}\!(s)\!-\!H_N(s),\!\sigma_N\!(s)])ds.
\end{equation}
Then, by (\ref{induction1}), we have
\begin{equation}\label{induction5}
    \begin{aligned}
    &e_N ((J+1)h)\\
    &=\mathcal{U}_{\rho_0}[(J+1)h,0](\rho_0-\sigma_0)+\!\int_{0}^{(J\!+\!1)h}\! \mathcal{U}_{\rho_0}[(J\!+\!1)h,\!s](-\text{i}[H_{\rho_0}\!(s)\!-\!H_N (s),\!\sigma_N\!(s)])ds.
    \end{aligned}
\end{equation}
Since $T=Nh$, it has been shown by mathematical induction that (\ref{induction2}) holds.

For $s \in [0,T)$, we denote
\begin{equation}
    G_N (s)\triangleq \mathcal{U}_{\rho_0}[T,\!s]\big(-\text{i}[H_{\rho_0}(s)-H_N (s),\sigma_N (s)]\big).
\end{equation}
Then, for $s\in [nh,(n+1)h)$, $0\leq n\leq N-1$, by Taylor's theorem with the Peano form of the remainder,
\begin{equation}\label{firstorder}
    G_N (s)= G_N(nh)+\dot{G}_N (nh^+)(s-nh)+R_{n,n+1}^N (s) (s-nh),
\end{equation}
where
\begin{equation}
    \lim_{s\to nh^+} R_{n,n+1}^N (s)=0.
\end{equation}
From (\ref{induction2}) and (\ref{firstorder}), we have
\begin{equation}\label{thm1main}
    \begin{aligned}
    &e_N (T)-\mathcal{U}_{\rho_0}[T,0](\rho_0-\sigma_0)\\
    &=\sum_{n=0}^{N-1}G_N(nh)h+\frac{1}{2}\dot{G}_N (nh^+)h^2+\int_{nh}^{(n+1)h}\!R_{n,n+1}^N (s) (s-nh)ds.
    \end{aligned}
\end{equation}
The terms on the r.h.s. of the last equation in (\ref{thm1main}) are analysed separately. Firstly,
\begin{equation}\label{thm1main1}
    \begin{aligned}
    \Vert\sum_{n=0}^{N-1}G_N(nh)h\Vert&=\Vert\sum_{n=0}^{N-1}\mathcal{U}_{\rho_0}[T,nh](-\text{i}[E(N,n),\sigma_N (nh)])\Vert h\\
    &\leq \sum_{n=0}^{N-1}\Vert\mathcal{U}_{\rho_0}[T,nh](-\text{i}[E(N,n),\sigma_N (nh)])\Vert h\\
    &=\sum_{n=0}^{N-1}\Vert[E(N,n),\sigma_N (nh)]\Vert h\\
    &\leq 2\sum_{n=0}^{N-1}\Vert E(N,n)\Vert h\\
    &\leq 2M_1 h T,
    \end{aligned}
\end{equation}
where we have applied the unitarity of superoperator $\mathcal{U}_{\rho_0}[\cdot,\cdot]$. From (\ref{thm1main1}), we have
\begin{equation}\label{thm1main1c}
    \lim_{N\to+\infty}\sum_{n=0}^{N-1}G_N(nh)h=0.
\end{equation}

Secondly, it is shown in Appendix A that the following equation holds:
\begin{equation}\label{appendixshown}
    \frac{d}{ds}\big(\mathcal{U}_{\rho_0}[T,s](A(s))\big)=\mathcal{U}_{\rho_0}[T,s]\big(\text{i}[H_{\rho_0}(s),A(s)]+\dot{A}(s)\big),
\end{equation}
where $A(\cdot)$ is a continuously differentiable matrix function. For $s \in \big(nh,(n+1)h\big)$, $0 \leq n \leq N-1$, we derive that:
\begin{equation}\label{help26}
    \begin{aligned}
    \frac{d}{ds}\big(-\text{i}[H_{\rho_0}(s)-H_N^n,\sigma_N (s)]\big)&=-\text{i}\big([\dot{H}_{\rho_0}(s),\sigma_N (s)]+[H_{\rho_0}(s)-H_N^n,\dot{\sigma}_N (s)]\big)\\
    &=-\text{i}[\dot{H}_{\rho_0}(s),\sigma_N (s)]-\big[H_{\rho_0}(s)-H_N^n,[H_N^n,\sigma_N (s)]\big].
    \end{aligned}
\end{equation}

With the aid of (\ref{appendixshown}) and (\ref{help26}), we arrive at:
\begin{multline}\label{1derivative}
    \dot{G}_N (nh^+)=\mathcal{U}_{\rho_0}[T,nh]\bigg(\big[H_{\rho_0}(nh),[E(N,n),\sigma_N (nh)]\big]\\
    -\text{i}[\dot{H}_{\rho_0}(nh),\sigma_N (nh)]-\big[E(N,n),[H_N^n,\sigma_N (nh)]\big]\bigg).
\end{multline}
Since we have assumed the second-order continuous differentiability of feedback protocols, provided with boundedness of the system trajectory, there exist $P_0 >0$ and $P_1 >0$, such that $\Vert H_{\rho_0}(s)\Vert \leq P_0$ and $\Vert \dot{H}_{\rho_0}(s) \Vert \leq P_1$ hold for $s \geq 0$. Moreover, $\forall N \in \mathbb{N}^+$ and $0 \leq n \leq N-1$, 
\begin{equation}\label{28}
    \begin{aligned}
    \Vert H_N^n \Vert&=\Vert H_N^n-H_{\rho_0}(nh)+H_{\rho_0}(nh) \Vert \leq \Vert H_N^n-H_{\rho_0}(nh)\Vert+\Vert H_{\rho_0}(nh)\Vert\\
                      &\leq M_1 h+P_0.
    \end{aligned}
\end{equation}
Next, it is clear that   
\begin{equation}\label{thm1main2}
    \begin{aligned}
    &\Vert\sum_{n=0}^{N-1}\dot{G}_N (nh^+)h^2\Vert\\
    &\leq \sum_{n=0}^{N-1}\big(4\Vert H_{\rho_0}(nh) \Vert\cdot \Vert E(N,n) \Vert + 2 \Vert \dot{H}_{\rho_0}(nh) \Vert+4\Vert H_N^n \Vert\cdot \Vert E(N,n) \Vert \big)h^2\\
    &\leq \sum_{n=0}^{N-1} \big(4P_0 M_1 h+2 P_1+ 4(M_1h +P_0)M_1 h\big)h^2\\
    &=8P_0 M_1 T h^2+ 2P_1 Th +4M_1^2 T h^3.
    \end{aligned}
\end{equation}
From (\ref{thm1main2}), it is clear that 
\begin{equation}\label{thm1main2c}
   \textcolor{blue}{ \lim_{N \to +\infty} \sum_{n=0}^{N-1}\frac{1}{2}\dot{G}_N (nh^+\!)h^2=0.}
\end{equation}

Thirdly, it follows from direct computation that
\begin{equation}\label{remainder}
    R_{n,n+1}^N (s)=\frac{G_N (s)-G_N (nh)}{s-nh}+\dot{G}_N (nh^+),
\end{equation}
where $s \in (nh,nh+h)$, $0 \leq n \leq N-1$. Also, $\exists s_{ij}^{N,n} \in \mathbb{R}$, $1 \leq i,j \leq N_S$, such that
\begin{equation}\label{u11}
    \Bigg(\frac{G_N (s)-G_N (nh)}{s-nh}\Bigg)_{ij}=\dot{(G_N)_{ij}}|_{s_{ij}^{N,n}},
\end{equation}
$\forall N \in \mathbb{N}^+$, $0 \leq n \leq N-1$, $s \in \big(nh, (n+1)h\big)$, with $nh < s_{ij}^N \leq s$.
From (\ref{1derivative}), it is derived that
\begin{equation}\label{u12}
    \Vert\dot{G}_N (nh^+)\Vert \leq 8P_0 M_1 T+ 2P_1 +4(M_1 T)^2,
\end{equation}
$\forall N \in \mathbb{N}^+$, $0 \leq n \leq N-1$. Also, from (\ref{appendixshown}) and (\ref{28}), we arrive at
\begin{equation}\label{u13}
    \Vert\dot{G}_N (s)\Vert \leq 4(M_1 T + 2P_0)^2 +2 P_1,
\end{equation}
$\forall N \in \mathbb{N}^+$, $0 \leq n \leq N-1$, and $s \in \big(nh, (n+1)h\big)$. 

Combining (\ref{u11}), (\ref{u12}), and (\ref{u13}), it is seen that there exists $R>0$, such that
\begin{equation}\label{u14}
    \Vert R_{n,n+1}^N (s)\Vert \leq R,
\end{equation}
$\forall N \in \mathbb{N}^+$, $0 \leq n \leq N-1$, and $s \in \big(nh, (n+1)h\big)$. Consequently, it holds that
\begin{equation}\label{thm1main3}
    \begin{aligned}
    \Vert \sum_{n=0}^{N-1}\int_{nh}^{nh+h}R_{n,n+1}^N (s)(s-nh)ds\Vert &\leq \sum_{n=0}^{N-1}\Vert \int_{0}^{h}R_{n,n+1}^N (u+nh)udu  \Vert\\
    &\leq \sum_{n=0}^{N-1} \int_{0}^{h} \Vert R_{n,n+1}^N (u+nh) \Vert u du\\
    &\leq \frac{R}{2}\sum_{n=0}^{N-1} h^2=\frac{R}{2}hT,
    \end{aligned}
\end{equation}
which leads to:
\begin{equation}\label{thm1main3c}
    \lim_{N \to +\infty} \sum_{n=0}^{N-1}\int_{nh}^{nh+h}R_{n,n+1}^N (s)(s-nh)ds=0.
\end{equation}

Due to (\ref{thm1main1c}), (\ref{thm1main2c}), and (\ref{thm1main3c}), eq. (\ref{s1}) holds, which completes the proof. $\hfill\square$
\end{proof}

\newtheorem{rmk}{Remark}
\begin{rmk}
From Theorem 1, we see that, \textcolor{blue}{if the chosen numerical method achieves errors upper bounded by $O(h)$ in terms of approximating the time-varying Hamiltonian,} then the norm of $e_N (T)$ converge to the norm of initial error $\rho_0 -\sigma_0$ as the number of simulation grids tends to infinity.
\end{rmk}
We proceed to the next result of this article.

\begin{thm}
$\forall T >0$ fixed, if $\exists M_2 >0$, such that \textcolor{blue}{$\Vert E(N,n)\Vert \leq \frac{M_2 T^2}{N^2}$}, $\forall N \in \mathbb{N}^+$ and $0 \leq n \leq N-1$, then the following equation holds:
\begin{multline}\label{s2}
    \lim_{N \to +\infty} \frac{e_N (T)-\mathcal{U}_{\rho_0}[T,0](\rho_0 -\sigma_0)}{\frac{T}{N}}\\
    =\frac{1}{2}\int_0^T \mathcal{U}_{\rho_0}[T,s]\big(-\text{i}[\dot{H}_{\rho_0}(s),\sigma (s)]\big )ds,
\end{multline}
where
\begin{equation}\label{sigmas}
    \sigma (s)=\mathcal{U}_{\rho_0}[s,0](\sigma_0),\quad s \geq 0.
\end{equation}
\end{thm}

\begin{proof}
$\forall N \in \mathbb{N}^+$, $0 \leq n \leq N-1$, and $nh \leq s < (n+1)h$, we have the following equation by Taylor's theorem with the Peano form of the remainder (this time to the second order):
\begin{multline}\label{secondorder}
    G_N (s)=G_N(nh)+\dot{G}_N(nh^+)(s-nh)\\
    +\frac{1}{2}\ddot{G}_N(nh^+) (s-nh)^2+Q_{n,n+1}^N (s)(s-nh)^2,
\end{multline}
where 
\begin{equation}
    \lim_{s\to nh}Q_{n,n+1}^N (s)=0.
\end{equation}
From (\ref{induction2}), we see that
\begin{equation*}
 \frac{e_N (T)-\mathcal{U}_{\rho_0}[T,0](\rho_0 -\sigma_0)}{h}=\frac{\int_0^T G_N (s) ds}{h}=\sum_{n=0}^{N-1} \frac{\int_{nh}^{nh+h} G_N (s) ds}{h}. 
\end{equation*}
Based on (\ref{secondorder}), we have 
\begin{equation}\label{secondorder1}
\begin{aligned}
    \sum_{n=0}^{N-1} \frac{\int_{nh}^{nh+h} G_N (s) ds}{h}&=\sum_{n=0}^{N-1} \Big(G_N(nh)+\frac{1}{2}\dot{G}_N(nh^+)h\\
    &+\frac{1}{3!}\ddot{G}_N(nh^+)h^2+\frac{\int_0^h Q_{n,n+1}^N (u+nh)u^2du}{h}\Big).
    \end{aligned}
\end{equation}

The r.h.s. of (\ref{secondorder1}) are analysed separately. Firstly,
\begin{equation}\label{thm2main1}
    \begin{aligned}
        \Vert \sum_{n=0}^{N-1} G_N (nh)\Vert &=\Vert \sum_{n=0}^{N-1}\mathcal{U}_{\rho_0}[t,nh]\big(-\text{i}[E(N,n),\sigma_N (nh)]\big) \Vert\\
        &\leq \sum_{n=0}^{N-1} 2\Vert E(N,n)\Vert \\
        &\leq 2\sum_{n=0}^{N-1} M_2 h^2=2M_2 h T.
    \end{aligned}
\end{equation}
Based on (\ref{thm2main1}), we have
\begin{equation}\label{thm2main1c}
    \lim_{N \to +\infty}\sum_{n=0}^{N-1} G_N (nh)=0.
\end{equation}

Secondly, taking the boundedness of $H_{\rho_0}(s)$ on $[0,+\infty)$ into account, we have
\begin{equation}\label{thm2main21}
    \begin{aligned}
        &\Vert \sum_{n=0}^{N-1} \mathcal{U}_{\rho_0}[T,nh]\big(H_{\rho_0}(nh),[E(N,n),\sigma_N (nh)]\big)h\Vert \\
        &\leq 4\sum_{n=0}^{N-1} \Vert H_{\rho_0}(nh)\Vert \cdot \Vert E(N,n)\Vert h \\
        &\leq 4\sum_{n=0}^{N-1} P_0 M_2 h^3=4 P_0 M_2 h^2 T.
    \end{aligned}
\end{equation}
Therefore, it is true that
\begin{equation}\label{thm2main21c}
    \lim_{N \to +\infty}\sum_{n=0}^{N-1} \mathcal{U}_{\rho_0}[T,nh]\bigg(\big[H_{\rho_0}(nh),[E(N,n),\sigma_N (nh)]\big]\bigg)h=0.
\end{equation}
Also, similar to (\ref{28}), we have
\begin{equation}\label{28another}
         \Vert H_N^n\Vert \leq M_2 h^2 +P_0. 
\end{equation}
With the aid of (\ref{28another}), it is shown that
\begin{equation}\label{thm2main22}
    \begin{aligned}
        &\Vert \sum_{n=0}^{N-1} \mathcal{U}_{\rho_0}[T,nh]\big(E(N,n),[H_N^n,\sigma_N (nh)]\big)h\Vert \\
        &\leq 4\sum_{n=0}^{N-1}\Vert H_N^n\Vert \cdot \Vert E(N,n)\Vert h \\
        &\leq 4 M_2 (M_2 h^4 T + P_0 h^2 T).
    \end{aligned}
\end{equation}
As a result of (\ref{thm2main22}), it holds that
\begin{equation}\label{thm2main22c}
    \lim_{N \to +\infty}\sum_{n=0}^N  \mathcal{U}_{\rho_0}[T,nh]\bigg(\big[E(N,n),[H_N^n,\sigma_N (nh)]\big]\bigg)h=0.
\end{equation}
Next, based on (\ref{thm1main1}), (\ref{thm1main2}), and (\ref{thm1main3}), we arrive at
\begin{equation}\label{thm2aid1}
    \begin{aligned}
        \Vert \sigma (nh)-\sigma_N (nh) \Vert &=\Vert \rho(nh;\rho_0,0) -\sigma_N (nh) -(\rho(nh;\rho_0,0)-\sigma (nh)) \Vert\\
        &=\Vert e_N (nh)-\mathcal{U}_{\rho_0}[nh,0](\rho_0-\sigma_0) \Vert\\
        &\leq (2M_1+2P_1+\frac{1}{2}R)nh^2+8P_0 M_1 nh^3+4M_1^2 n h^4,
    \end{aligned}
\end{equation}
for $0 \leq n \leq N-1$. It thus follows that there exist three positive real numbers $B_1$, $B_2$, and $B_3$, such that
\begin{equation}\label{thm2aid2}
    \begin{aligned}
        \Vert \sum_{n=0}^{N-1}-\text{i}[\dot{H}_{\rho_0}(nh),\sigma_N (nh)-\sigma(nh)]h\Vert&\leq \sum_{n=0}^{N-1} B_1 nh^3 + B_2 n h^4 +B_3 nh^5\\
        &=B_1\frac{N-1}{2N^2}+B_2 h\frac{N-1}{2N^2}+B_3 h^2\frac{N-1}{2N^2}.
    \end{aligned}
\end{equation}
It thus follows that 
\begin{equation}\label{thm2aid3}
    \lim_{N\to+\infty}\sum_{n=0}^{N-1}\mathcal{U}_{\rho_0}\big(-\text{i}[\dot{H}_{\rho_0}(nh),\sigma_N (nh)-\sigma(nh)]\big)h=0.
\end{equation}
Since 
\begin{multline}
    \mathcal{U}_{\rho_0}[T,nh]\big(-\text{i}[\dot{H}_{\rho_0}(nh),\sigma_N(nh)]\big)\\=\mathcal{U}_{\rho_0}[T,nh]\big(-\text{i}[\dot{H}_{\rho_0}(nh),\sigma(nh)]-\text{i}[\dot{H}_{\rho_0}(nh),\sigma_N (nh)-\sigma(nh)]\big),
\end{multline}
from (\ref{thm2aid3}), it is true that
\begin{equation}\label{thm2main23c}
    \begin{aligned}
    &\lim_{N\to +\infty}\sum_{n=0}^{N-1}\mathcal{U}_{\rho_0}[T,nh]\big(-\text{i}[\dot{H}_{\rho_0}(nh),\sigma_N(nh)]\big)h\\
    &=\int_0^T \mathcal{U}_{\rho_0}[T,s]\big(-\text{i}[\dot{H}_{\rho_0}(s),\sigma(s)]\big)ds.
    \end{aligned}
\end{equation}

Thirdly, for $0 \leq n \leq N-1$, $s \in (nh, nh+h)$, we have
\begin{equation}\label{1thderivative}
\begin{aligned}
    \dot{G}_N (s)&=\mathcal{U}_{\rho_0}[T,s]\Big( \big[H_{\rho_0}(s),[H_{\rho_0}(s)-H_N^n,\sigma_N (s)]\big]\\
    &-\text{i}[\dot{H}_{\rho_0}(s),\sigma_N (s)]-\big[H_{\rho_0}(s)-H_N^n,[H_N^n,\sigma_N (s)]\big]\Big)\\
    &\triangleq \mathcal{U}_{\rho_0}[T,s]\big(B_N (s)\big).
    \end{aligned}
\end{equation}
By (\ref{appendixshown}), it is clear that
\begin{equation}
    \ddot{G}_N (s)=\mathcal{U}_{\rho_0}[T,s]\big(\text{i}[H_{\rho_0}(s),B_N (s)]+\dot{B}_N (s)\big).
\end{equation}
we evaluate that
\begin{equation}
    \Vert B_N (s) \Vert \leq 4(2P_0+M_2 T^2)^2 +2P_1,
\end{equation}
for $0 \leq n \leq N-1$, $s \in (nh, nh+h)$. Therefore, it is clear that $\exists R_1 >0$, such that
\begin{equation}\label{58}
    \Vert \mathcal{U}_{\rho_0}[T,s]\big(\text{i}[H_{\rho_0}(s),B_N (s)] \Vert \leq R_1,
\end{equation}
$\forall N \in \mathbb{N}^+$, $0 \leq n \leq N-1$, and $s \in (nh, nh+h)$. By applying (\ref{28another}), we have  
\begin{equation}\label{2derivative0nh}
    \Vert B_N (nh^+) \Vert \leq 8P_0 M_2 T^2+4(M_2 T^2)^2 +2P_1.
\end{equation}
We then proceed to evaluate $\dot{B}_N (s)$. It is computed that, for $0 \leq n \leq N-1$, $s \in (nh, nh+h)$,
\begin{multline}\label{2derivative1th}
    \frac{d}{ds}\Big(\big[H_{\rho_0}(s),[H_{\rho_0}(s)-H_N^n,\sigma_N (s)]\big]\Big)\\
    =\dot{H}_{\rho_0}(s)\big(H_{\rho_0}(s)-H_N^n\big)\sigma_N (s)+H_{\rho_0}(s)\dot{H}_{\rho_0}(s)\sigma_N (s)\\
    +H_{\rho_0}(s)\big(H_{\rho_0}(s)-H_N^n\big)\dot{\sigma}_N (s)\\
    -\dot{H}_{\rho_0}(s)\sigma_N (s)\big(H_{\rho_0}(s)-H_N^n\big)-H_{\rho_0}(s)\dot{\sigma}_N (s)\big(H_{\rho_0}(s)-H_N^n\big)\\
    -H_{\rho_0}(s)\sigma_N (s)\dot{H}_{\rho_0}(s)\\
    -\dot{H}_{\rho_0}(s)\sigma_N (s)H_{\rho_0}(s)-\big(H_{\rho_0}(s)-H_N^n\big)\dot{\sigma}_N (s)H_{\rho_0}(s)\\
    -\big(H_{\rho_0}(s)-H_N^n\big)\sigma_N (s)\dot{H}_{\rho_0}(s)\\
    +\dot{\sigma}_N (s)\big(H_{\rho_0}(s)-H_N^n\big)H_{\rho_0}(s)+\sigma_N (s)\dot{H}_{\rho_0}(s)H_{\rho_0}(s)\\
    +\sigma_N (s)\big(H_{\rho_0}(s)-H_N^n\big)\dot{H}_{\rho_0}(s).
\end{multline}
Consequently, we find that
\begin{multline}\label{2derivative1nh}
    \frac{d}{ds}\Big(\big[H_{\rho_0}(s),[H_{\rho_0}(s)-H_N^n,\sigma_N (s)]\big]\Big)\Big{|}_{nh^+}\\
    =\Big[H_{\rho_0}(nh^+),\big[E(N,n),-\text{i}[H_N^n,\sigma_N (nh)]\big]\Big]\\
    +\big[H_{\rho_0}(nh^+),[\dot{H}_{\rho_0}(nh^+),\sigma_N (nh)]\big]\\
    +\big[H_{\rho_0}(nh^+),[\dot{H}_{\rho_0}(nh^+),[E(N,n),\sigma_N (nh)]\big].
\end{multline}
Next, it is evaluated that, for $0 \leq n \leq N-1$, $s \in (nh, nh+h)$,
\begin{equation}\label{2derivative2th}
    \frac{d}{ds}\big(-\text{i}[\dot{H}_{\rho_0}(s),\sigma_N (s)]\big)=-\text{i}\big([\ddot{H}_{\rho_0}(s),\sigma_N (s)]+[\dot{H}_{\rho_0}(s),\dot{\sigma}_N (s)]\big).
\end{equation}
It follows that
\begin{multline}\label{2derivative2nh}
    \frac{d}{ds}\big(-\text{i}[\dot{H}_{\rho_0}(s),\sigma_N (s)]\big)\Big{|}_{nh^+}\\
    =-\text{i}\big([\ddot{H}_{\rho_0}(nh^+),\sigma_N (nh^+)]+[\dot{H}_{\rho_0}(nh^+),-\text{i}[H_N^n,\sigma_N (nh)]\big).
\end{multline}
Moreover, for $0 \leq n \leq N-1$, $s \in (nh, nh+h)$, 
\begin{multline}\label{2derivative3th}
    \frac{d}{ds}\Big(\big[[H_N^n,\sigma_N (s)],H_{\rho_0}(s)-H_N^n\big]\Big)\\
    =H_N^n\dot{\sigma}_N (s)(H_{\rho_0}(s)-H_N^n)+H_N^n\sigma_N (s)\dot{H}_{\rho_0}(s)\\
    -\dot{\sigma}_N (s)H_N^n(H_{\rho_0}(s)-H_N^n)-\sigma_N (s)H_N^n\dot{H}_{\rho_0}(s)\\
    -\dot{H}_{\rho_0}(s)H_N^n\sigma_N (s)-(H_{\rho_0}(s)-H_N^n)H_N^n\dot{\sigma}_N (s)\\
    +\dot{H}_{\rho_0}(s)\sigma_N (s)H_N^n+(H_{\rho_0}(s)-H_N^n)\dot{\sigma}_N (s)H_N^n.
\end{multline}
We also have
\begin{multline}\label{2derivative3nh}
    \frac{d}{ds}\Big(\big[[H_N^n,\sigma_N (s)],H_{\rho_0}(s)-H_N^n\big]\Big)\Big{|}_{nh^+}\\
    =\bigg[\big[H_N^n,-\text{i}[H_N^n,\sigma_N (s)\big],E(N,n)\bigg]+\big[[H_N^n,\sigma_N (nh)],\dot{H}_{\rho_0}(nh^+)\big].
\end{multline}
Based on (\ref{58})-(\ref{2derivative3nh}), we see that there exists $R_2 >0$, such that $\Vert \ddot{G}_N (s)\Vert \leq R_2$, $\forall N \in \mathbb{N}^+$, $0 \leq n \leq N-1$ and $nh<s<\big((n+1)h\big)$. As result, it is true that
\begin{equation}\label{thm2main3}
    \Vert \sum_{n=0}^{N-1}\ddot{G}_N (nh^+) \Vert h^2 \leq R_2\sum_{n=0}^{N-1}h^2 =R_2 h T,
\end{equation}
and that
\begin{equation}\label{thm2main3c}
     \lim_{N\to +\infty}\sum_{n=0}^{N-1}\ddot{G}_N (nh^+)  h^2 =0.
\end{equation}

Fourthly, from (\ref{secondorder}), it is derived that
\begin{equation}\label{remainderbounded}
\begin{aligned}
    Q_{n,n+1}^N (s)&=\frac{G_N (s)-G_N (nh)-\dot{G}_N (nh^+)(s-nh)}{(s-nh)^2}+\frac{1}{2}\ddot{G}_N(nh^+)\\
    &=\frac{\frac{G_N (s)-G_N (nh)}{s-nh}-\dot{G}_N (nh^+)}{(s-nh)}+\frac{1}{2}\ddot{G}_N(nh^+),
    \end{aligned}
\end{equation}
where $0 \leq n \leq N-1$ and $s \in (nh, nh+h)$. Let us denote
\begin{equation}
    \hat{G}_N(s)\triangleq \frac{G_N (s)-G_N (nh)}{s-nh},
\end{equation}
where $0 \leq n \leq N-1$ and $s \in (nh, nh+h)$. Then, for $1\leq i,j \leq N_S$, eq. (\ref{u11}) translates to
\begin{equation*}
\Big(\hat{G}_N(s)\Big)_{ij}=\dot{(G_N)_{ij}}|_{s_{ij}^{N,n}}.
\end{equation*}
Moreover, there exists $\theta_{ij}^{N,n}\in (nh,s_{ij}^{N,n}]$, such that 
\begin{equation}\label{u21}
    \bigg(\frac{\dot{(G_N)_{ij}}|_{s_{ij}^{N,n}}-\dot{(G_N)_{ij}}|_{nh^+}}{s_{ij}^{N,n}-nh}\bigg)=\ddot{\big(G_N\big)}_{ij}\big{|}_{\theta_{ij}^{N,n}} \bigg(\frac{s_{ij}^{N,n}-nh}{s-nh}\bigg)
\end{equation}
Provided with (\ref{u21}) and the boundedness of $\ddot{G}_N$ already mentioned, it is true that there exists $W>0$, such that
\begin{equation}
    \Vert Q_{n,n+1}^N (s)\Vert \leq W,
\end{equation}
$\forall N \in \mathbb{N}^+$, $0 \leq n \leq N-1$ and $nh<s<\big((n+1)h\big)$. It thus follows that
\begin{equation}\label{thm2main4}
    \frac{\Vert\sum_{n=0}^{N-1}\int_0^h Q_{n,n+1}^N (u+nh)u^2du\Vert}{h}\leq \frac{\sum_{n=0}^{N-1}W\int_0^h u^2du}{h}=\frac{1}{3}WhT.
\end{equation}
Consequently, 
\begin{equation}\label{thm2main24c}
    \lim_{N\to +\infty}\frac{\sum_{n=0}^{N-1}\int_0^h Q_{n,n+1}^N (u+nh)u^2du}{h}=0.
\end{equation}

Finally, combining (\ref{secondorder1}), (\ref{thm2main1c}), (\ref{thm2main21c}), (\ref{thm2main22c}), (\ref{thm2main23c}), (\ref{thm2main3c}) and (\ref{thm2main24c}), eq. (\ref{s2}) follows. The proof is completed. $\hfill\square$
\end{proof}

\begin{rmk}
From Theorem 2, we know that, \textcolor{blue}{provided that the r.h.s of (\ref{s2}) is not zero and the numerical algorithm achieves errors upper bounded by $O(h^2)$ in terms of approximating the time-varying Hamiltonian}, the convergence of $e_N(T)$ to $\mathcal{U}_{\rho_0}[T,0](\rho_0-\sigma_0)$ as $N\to +\infty$ is dominated by the first power of $h$. Moreover, it should be noted that the r.h.s of (\ref{s2}) is \textbf{independent of} the numerical approximation method, which implies that adopting more accurate (thus possibly more expensive) numerical methods may \textbf{not} lead to more efficient error convergence, once the accuracy reaches above a certain level. This is a key finding of our work, which may provide guidance on saving computational resources.
\end{rmk}

\subsection{Error Upper Bound}
In this subsection, we present an upper bound estimation of the error. The following notations are in order. 

Let $\{e_j\}_{j=1}^{N_S}$ be a set of $N_S$-dimensional real column vectors, with $e_j(k)=\delta_{jk}$, $1 \leq j,k \leq N_S$. Moreover, let us define $E_{jk}\triangleq e_j e_k^T+e_k e_j^T$, $F_{jk}\triangleq \text{i}e_j e_k^T-\text{i}e_k e_j^T$, $1\leq j<k \leq N_S$, and $\Omega_{kk}\triangleq e_k e_k^T$, $1\leq k \leq N_S$. Then, $\forall \rho \in \mathcal{D}(\mathcal{H})$, there exists $N_S^2-1$ real numbers, which are denoted as:
\begin{equation}
    \{\rho_{ii}|1\leq i \leq N_S -1\}\cup\{\rho_{ij}^R|1\leq i<j\leq N_S\}\cup\{\rho_{ij}^I|1\leq i<j\leq N_S\},
\end{equation}
such that 
\begin{equation}
    \rho=\sum_{i=1}^{N_S-1}\rho_{ii}\Omega_{ii}+\sum_{1\leq i<j\leq N_S}(\rho_{ij}^R E_{ij}+\rho_{ij}^I F_{ij}).
\end{equation}

For each $k \in \{1,...,M\}$, the feedback protocol $u_k (\cdot)$ is viewed as a real function $\hat{u}_k (\cdot)$ of $N_S -1$ real variables:
\begin{equation}\label{notation3_1}
    \hat{u}_k\big(\{\rho_{ii}\}_{i=1}^{N_S-1};\{\rho_{ij}^R\}_{1\leq i<j\leq N_S};\{\rho_{ij}^I\}_{1\leq i<j\leq N_S}\big)\\
    \triangleq u_k(\rho), \quad \rho \in \mathcal{D}(\mathcal{H}).
\end{equation}
Next, we denote:
\begin{equation}\label{notation3_2}
    \hat{u}_{ii}^k\triangleq \frac{\partial \hat{u}_k }{\partial \rho_{ii}}, \quad 1\leq i \leq N_S -1,
\end{equation}
and
\begin{equation}\label{notation3_3}
    \hat{u}_{ij}^{k,R}\triangleq \frac{\partial \hat{u}_k }{\partial \rho_{ij}^R}, \hat{u}_{ij}^{k,I}\triangleq \frac{\partial \hat{u}_k }{\partial \rho_{ij}^I}, 1\leq i <j \leq N_S,
\end{equation}
where $1\leq k\leq M$.

Our result on error upper bound is given as follows.

\begin{thm}
$\forall T>0$, the following inequality holds:
\begin{multline}\label{th3}
    \Vert e_N (T)\Vert \leq \Vert \rho_0 -\sigma_0 \Vert + 2\sum_{n=0}^{N-1}\Vert E(N,n) \Vert \frac{T}{N}\\
    +2 \sum_{k=1}^M L_{u_k}^1\big( \Vert H_0\Vert +\sum_{j=1}^M L_{u_j}^0\Vert H_j \Vert \big)\Vert H_k \Vert \frac{T^2}{N},
\end{multline}
where $L_{u_j}^0 \triangleq \max_{\rho \in \{\rho(t;\rho_0,0)\}_{t=0}^{+\infty}} \vert u_j (\rho)\vert$, and 
\begin{equation}
    L_{u_j}^1 \triangleq \\
    \max_{\rho \in \mathcal{D}(\mathcal{H})}\sqrt{\sum_{i=1}^{N_S-1}(\hat{u}_{ii})^2 + \sum_{1\leq i<k \leq N_S}\frac{(\hat{u}_{ik}^{j,R})^2+(\hat{u}_{ik}^{j,I})^2}{2}}, 
\end{equation}
for $1\leq j \leq M$.
\end{thm}

\begin{proof}
Based on (\ref{induction1}), we see that
\begin{equation}\label{thm3_1}
\begin{aligned}
    &\Vert e_N ((n+1)h) \Vert \\ 
    &\leq  \Vert(e_N(nh)) \Vert+\int_{nh}^{(n+1)h}  \Vert(-\text{i}[H_{\rho_0}(s)-H_N^n,\sigma_N(s)] \Vert ds\\
    &\leq \Vert(e_N(nh)) \Vert +2 \int_{nh}^{(n+1)h}\Vert H_{\rho_0}(s)-H_{\rho_0}(nh) \Vert + \Vert H_{\rho_0}(s)-H_{\rho_0}(nh) \Vert ds\\
    &=\Vert(e_N(nh)) \Vert +2 \int_{nh}^{(n+1)h}\Vert H_{\rho_0}(s)-H_{\rho_0}(nh) \Vert + \Vert E(N,n) \Vert ds,
    \end{aligned}
\end{equation}
for $1 \leq n \leq N-1$.

On one hand, it is clear that, for $1 \leq n \leq N-1$,
\begin{equation}\label{thm3_2}
    \int_{nh}^{(n+1)h}\Vert E(N,n) \Vert ds=\Vert E(N,n) \Vert h.
\end{equation}
On the other hand, for $1 \leq n \leq N-1$ and $s\in (nh,nh+h)$,
\begin{multline}\label{thm3_3}
     H_{\rho_0}(s)-H_{\rho_0}(nh)=\sum_{k=1}^M \Big(u_k\big(\rho(s;\rho_0,0)\big)-u_k\big(\rho(nh;\rho_0,0)\big)\Big)H_k\\
     =\sum_{k=1}^M \Big(\frac{d u_k\big(\rho(s;\rho_0,0)\big)}{ds} \Big{|}_{s_N^{n,k}}\Big)(s-nh)H_k,
\end{multline}
where $nh <s_N^{n,k} \leq s$, for $1 \leq k \leq M$.

Let us refer to (\ref{notation3_1}), (\ref{notation3_2}), and (\ref{notation3_3}). Then, it is derived that, for $s>0$ and $1 \leq k \leq M$,
\begin{multline}\label{thm3_4}
    \frac{d u_k\big(\rho(s;\rho_0,0)\big)}{ds}=\sum_{i=1}^{N_S-1}\hat{u}_{ii}^k \frac{d \rho_{ii}}{ds}+\sum_{1\leq i <j \leq N_S}\Bigg( \hat{u}_{ij}^{k,R} \frac{d \rho_{ij}^{k,R}}{ds}+ \hat{u}_{ij}^{k,I} \frac{d \rho_{ij}^{k,I}}{ds}\Bigg).
\end{multline}
Next, for each $k \in \{1,...,M\}$, we define an $N_S \times N_S$ matrix $D_{u_k}$, where
\begin{equation}\label{thm3_5}
    \begin{aligned}
    &(D_{u_k})_{ii}\triangleq \hat{u}_{ii},\quad 1\leq i \leq N_S-1;\\
    &(D_{u_k})_{N_S N_S}\triangleq 0;\\
    &(D_{u_k})_{jk}\triangleq \frac{\hat{u}_{jk}^R +\text{i}\hat{u}_{jk}^I}{2},\quad 1\leq j<k\leq N_S;\\
    &(D_{u_k})_{jk}\triangleq \frac{\hat{u}_{kj}^R -\text{i}\hat{u}_{kj}^I}{2},\quad 1\leq k<j\leq N_S.
    \end{aligned}
\end{equation}
In Appendix B, it is proved that
\begin{equation}\label{thm3_6}
   \frac{d u_k\big(\rho(s;\rho_0,0)\big)}{ds}=\text{tr}\big(D_{u_k}(-\text{i})[H_{\rho_0}(s),\rho(s;\rho_0,0)]\big),
\end{equation}
for $1 \leq k \leq M$ and $s>0$, with the aid of (\ref{thm3_4}). Then, we apply the Cauchy-Schwarz inequality and yield that 
\begin{equation}\label{thm3_7}
    \begin{aligned}
    \Bigg{|}\frac{d u_k\big(\rho(s;\rho_0,0)\big)}{ds}\Bigg{|}&\leq \Vert D_{u_k} \Vert \cdot \Vert (-\text{i})[H_{\rho_0}(s),\rho(s;\rho_0,0)] \Vert\\
                                                              &\leq 2 \Vert D_{u_k} \Vert \cdot \Vert H_{\rho_0}(s) \Vert,
    \end{aligned}
\end{equation}
for $s>0$. 

Then, $\Vert D_{u_k} \Vert$ and $\Vert H_{\rho_0}(s) \Vert$ are analysed separately. Firstly, we have
\begin{equation}\label{thm3_8}
    \begin{aligned}
    \Vert D_{u_k} \Vert &=\big(\text{tr}(D_{u_k}^\dagger D_{u_k})\big)^{\frac{1}{2}}=\big(\sum_{i=1}^{N_S}\sum_{j=1}^{N_S}\vert D_{u_k} \vert^2 \big)^{\frac{1}{2}}\\
                        &=\Bigg(\sum_{i=1}^{N_S-1}(\hat{u}_{ii})^2 + \sum_{1\leq i<k \leq N_S}\frac{(\hat{u}_{ik}^{j,R})^2+(\hat{u}_{ik}^{j,I})^2}{2}\Bigg)^{\frac{1}{2}}=L_{u_k}^1
    \end{aligned}
\end{equation}
Secondly, $\forall s>0$, it holds that 
\begin{equation}\label{thm3_9}
    \begin{aligned}
    \Vert H_{\rho_0}(s) \Vert &=\Vert H_0  + \sum_{k=1}^M u_k\big(\rho(s;\rho_0,0)\big) H_k \Vert\\
                              &\leq \Vert H_0 \Vert + \sum_{k=1}^M \vert u_k\big(\rho(s;\rho_0,0)\big)\vert \cdot \Vert H_k \Vert \\
                              &\leq \Vert H_0 \Vert + \sum_{k=1}^M L_{u_k}^0\Vert H_k \Vert.
    \end{aligned}
\end{equation}
Combining (\ref{thm3_1}), (\ref{thm3_2}), (\ref{thm3_3}), (\ref{thm3_7}), (\ref{thm3_8}) and (\ref{thm3_9}), we arrive at
\begin{multline}\label{thm3_10}
    \Vert e_N ((n+1)h) \Vert \leq  \Vert(e_N(nh)) \Vert + 2\Vert E(N,n) \Vert h\\
    +4\int_{nh}^{nh+h} \sum_{k=1}^M L_{u_k}^1 (\Vert H_0 \Vert + \sum_{k=1}^M L_{u_j}^0\Vert H_j \Vert)\Vert H_k \Vert(s-nh) ds\\
    = \Vert(e_N(nh)) \Vert + 2\Vert E(N,n) \Vert h \\
    +2\sum_{k=1}^M L_{u_k}^1 (\Vert H_0 \Vert + \sum_{k=1}^M L_{u_j}^0\Vert H_j \Vert)\Vert H_k \Vert h^2,
\end{multline}
for $0 \leq n \leq N-1$. 

Finally, eq. (\ref{th3}) follows from the summation of (\ref{thm3_10}) from $n=0$ to $n=N-1$. $\hfill\square$
\end{proof}

\begin{rmk}
The error upper bound specified in Theorem 3 involves the contribution of three factors. The first one is the norm of the initial error, the second one is the error in approximating the time-varying Hamiltonian, and the third one is a term proportional to $\frac{T^2}{N}$. It is noted that the coefficient before $\frac{T^2}{N}$ is \textbf{independent} of the numerical simulation method we choose. Therefore, it is not possible to diminish the influence of the third factor on the upper bound by adopting more accurate simulation methods.
\end{rmk}

\section{Example}
In this section, we give an example to illustrate our results. Let us set $N_S=4$, then $\mathcal{H}$ corresponds to the underlying Hilbert space of two qubits. We denote
\begin{equation}
    |1\rangle \triangleq \begin{pmatrix}1\\0 \end{pmatrix}, \quad |0\rangle \triangleq \begin{pmatrix}0\\1 \end{pmatrix}.
\end{equation}
It follows that $\mathcal{H}$ is spanned by the following orthonormal basis:
\begin{equation}
    \{|11\rangle,|10\rangle,|01\rangle,|00\rangle\}.
\end{equation}
Next, let us denote 
\begin{equation}
    X \triangleq \begin{pmatrix}0&1\\1&0 \end{pmatrix}, Y \triangleq \begin{pmatrix}0&-\text{i}\\ \text{i}&0 \end{pmatrix}, Z \triangleq \begin{pmatrix}1&0\\0&-1 \end{pmatrix}.
\end{equation}

Suppose that our goal is to prepare the following target state:
\begin{equation}
    \rho_d \triangleq |\Phi\rangle\langle \Phi|,\quad |\Phi\rangle \triangleq \frac{1}{\sqrt{2}}(|10\rangle+|01\rangle),
\end{equation}
which is an entangled state. In the closed-loop design stage, the following dynamics is proposed to achieve our goal:
\begin{equation}\label{eg1}
    \dot{\rho}=-\text{i}[H_0+u_1(\rho)H_1,\rho],\quad \rho(0)=\rho_0,
\end{equation}
where
\begin{equation}\label{eg2}
    \begin{aligned}
    &H_0=Z\otimes I_2+I_2\otimes Z;\\
    &H_1=X\otimes Y-Y\otimes X;\\
    &u_1(\cdot)=-K \text{tr}\big(\text{i}[\rho_d,H_1](\cdot)\big),\quad K>0.
    \end{aligned}
\end{equation}
The following proposition is presented and proved.
\newtheorem{prop}{Proposition}
\begin{prop}
Consider system (\ref{eg1}) with feedback protocol specified in (\ref{eg2}). Let $\rho_0=|10\rangle\langle 10|$, then it holds that
\begin{equation}
    \lim_{t\to +\infty} \rho(t;\rho_0,0)=\rho_d.
\end{equation}
\end{prop}

\begin{proof}
Consider the following Lyapunov function:
\begin{equation}\label{eg3}
    V(\rho)=\text{tr}(\Pi_d \rho),
\end{equation}
where $\Pi_d=I_4-\rho_d$. Next, we yield that
\begin{equation}\label{eg4}
    \dot{V}(\rho)=\text{tr}\big(\Pi_d(-\text{i})[H_0,\rho]\big)+u_1(\rho)\text{tr}\big(\Pi_d(-\text{i})[H_1,\rho]\big)
\end{equation}
On one hand, we have
\begin{equation}\label{eg5}
    \text{tr}(\Pi_d[H_0,\rho])=\text{tr}(\Pi_d H_0 \rho -\Pi_d \rho H_0)=\text{tr}(\Pi_d H_0 \rho -H_0 \Pi_d \rho)=\text{tr}\big([\Pi_d,H_0]\rho\big)=0.
\end{equation}
On the other hand, 
\begin{equation}\label{eg6}
    \text{tr}\big(\Pi_d(-\text{i})[H_1,\rho]\big)=\text{tr}\big(-\text{i}[\Pi_d,H_1]\rho\big)=\text{tr}\big(\text{i}[\rho_d,H_1]\rho\big).
\end{equation}
Combining (\ref{eg4}), (\ref{eg5}) and (\ref{eg6}), it follows that
\begin{equation}\label{eg7}
    \dot{V}(\rho)=-K\Big(\text{tr}\big(\text{i}[\rho_d,H_1]\rho\big)\Big)^2 \leq 0.
\end{equation}
By LaSalle's invariance principle, system trajectory $\rho(t;\rho_0,0)$ should converge to the largest invariant set contained in $\{\rho \in \mathcal{D}(\mathcal{H})|\dot{V}(\rho)=0\}$. Next, $\dot{V}(\rho)=0$ if and only if $\text{tr}\big(\text{i}[\rho_d,H_1]\rho\big)=0$, which is in turn equivalent to:
\begin{equation}\label{eg8}
    \langle 10|\rho|10\rangle= \langle 01|\rho|01\rangle.
\end{equation}
We further yield that
\begin{equation}\label{eg9}
    \begin{aligned}
    \langle 11|\dot{\rho}|11\rangle &=-\text{i}\text{tr}\big([H_0,\rho]\Pi_{11}\big)-\text{i}u_1 (\rho)\text{tr}\big([H_1,\rho]\Pi_{11}\big)\\
                                    &=-\text{i}\text{tr}\big([\Pi_{11},H_0]\rho\big)-\text{i}u_1 (\rho)\text{tr}\big([\Pi_{11},H_1]\rho\big)=0,
    \end{aligned}
\end{equation}
where $\Pi_{11}=|11\rangle\langle 11 |$. Similarly, we also yield that
\begin{equation}\label{eg10}
    \langle 00|\dot{\rho}|00\rangle =0.
\end{equation}
Since $\langle 11|\rho_0|11\rangle=0$ and $\langle 00|\rho_0|00\rangle=0$, it holds that
\begin{equation}\label{eg11}
    \begin{aligned}
    &\langle 11|\rho(t;\rho_0,0)|11\rangle =0;
    &\langle 00|\rho(t;\rho_0,0)|00\rangle =0,
    \end{aligned}
\end{equation}
for $t \geq 0$. Consequently, by semi-positivity, it is true that
\begin{equation}\label{eg12}
    \text{supp}\big(\rho(t;\rho_0,0)\big)=\{|10\rangle,|01\rangle\},
\end{equation}
for $t \geq 0$. Moreover, we have
\begin{equation}\label{eg13}
    \begin{aligned}
     \langle 10|\dot{\rho}|01\rangle&=-\text{i}\text{tr}\big([H_0,\rho]|01\rangle\langle 10|\big)-\text{i}u_1 (\rho)\text{tr}\big([H_1,\rho]|01\rangle\langle 10|\big)\\
                                    &=-\text{i}\text{tr}\big([|01\rangle\langle 10|,H_0]\rho\big)-\text{i}u_1 (\rho)\text{tr}\big([|01\rangle\langle 10|,H_1]\rho\big)\\
                                    &=-\text{i}u_1 (\rho)\text{tr}\big([|01\rangle\langle 10|,H_1]\rho\big)\\
                                    &=K(-2\langle 10|\rho|10\rangle+2\langle 01|\rho|01\rangle)^2\geq 0.
    \end{aligned}
\end{equation}
It is also true that
\begin{equation}\label{eg14}
    \langle 10|\dot{\rho}|01\rangle=\langle 01|\dot{\rho}|10\rangle \geq 0.
\end{equation}
Because $\langle 10|\rho_0|01\rangle=0$, and $\langle 01|\rho_0|10\rangle=0$, it is clear that
\begin{equation}\label{eg15}
    \begin{aligned}
    &\langle 10|\rho(t;\rho_0,0)|01\rangle\geq 0;
    &\langle 01|\rho(t;\rho_0,0)|10\rangle\geq 0,
    \end{aligned}
\end{equation}
for $t \geq 0$. 

The set $\mathcal{S}$ to which the system trajectory converges is restricted by (\ref{eg8}), (\ref{eg12}) and (\ref{eg15}). That is, $\forall \tilde{\rho} \in \mathcal{S}$, it should hold that
\begin{equation}\label{eg16}
    \begin{aligned}
    &\langle 10|\tilde{\rho}|10\rangle= \langle 01|\tilde{\rho}|01\rangle;\\
    &\text{supp}(\tilde{\rho})=\{|10\rangle,|01\rangle\};\\
    &\langle 10|\tilde{\rho}|01\rangle\geq 0;\\
    &\langle 01|\tilde{\rho}|10\rangle\geq 0,
    \end{aligned}
\end{equation}
We also note that $\frac{d}{dt}\text{tr}(\rho^2)=0$ holds for $t \geq 0$. Therefore, the system trajectory, which starts from pure state $\rho_0$, remains to be pure for $t \geq 0$. There is only one pure density operator, namely $\rho_d$, which satisfies (\ref{eg16}). As a result, the system trajectory must converge to $\rho_d$. $\hfill\square$
\end{proof}

Based on Proposition 1, we know that the trajectory of system (\ref{eg1}) starting from $\rho_0$ approaches $\rho_d$ in the long time limit. Such a trajectory, if realized in practice, would fulfill the task of preparing $\rho_d$. However, implementing (\ref{eg1}) with feedback protocol $u_1(\cdot)$ requires the real-time state. According to the procedures of closed-loop designed open-loop control, the dynamics of system (\ref{eg1}) is then simulated on a digital computer, yielding $\{\theta_j\}_{j=0}^N$ as an approximation. 

We move on to the open-loop control stage. The control input to be applied to the real system is:
\begin{equation}\label{eg17}
    \tilde{u}_1 (t)=u_1 (\theta_n),\quad t\in \Bigg[\frac{nT}{N},\frac{(n+1)T}{N}\Bigg).
\end{equation}
Next, we prepare an initial density operator $\sigma_0$ as close to $\rho_0$ as possible (perhaps by exciting one of the qubits and cooling the other, in terms of $|10\rangle\langle 10|$ in this example). The real system, which is subject to open-loop control $\tilde{u}$, is modelled by the following equation:
\begin{equation}\label{eg18}
    \sigma_N(t)=\sigma_0+\int_{0}^t -\text{i}[H_0+\tilde{u}_1(s)H_1,\sigma_N (s)]ds. 
\end{equation}

Several simulation results are then presented. Let us denote:
\begin{equation}\label{eg19}
    F(N,T) \triangleq e_N (T)-\mathcal{U}_{\rho_0}[T,0](\rho_0-\sigma_0).
\end{equation}
For this example, we set $T=1$ and $K=1$, and $\rho_0=|10\rangle\langle 10|$. Moreover, we assume that initialization is not perfect, setting $\sigma_0=0.95|10\rangle\langle 10|+0.05|00\rangle\langle 00|$. In terms of numerical approximation methods, we choose Runge-Kutta methods for approximation of (\ref{eg1}). Here, the $\text{S}^{\text{th}}$-order Runge-Kutta method is denoted as RKS and $\text{S} \in \{1,2,3,4,5\}$. Generally speaking, higher-order Runge-Kutta methods may achieve greater accuracy than lower-order ones in terms of approximating ordinary differential equations, at a price of being more complicated. We numerically simulate $F(N,1)$ for different values of $N$, applying RK1, RK2, RK3, RK4 and RK5. The simulated relationship between $F(N,1)$'s norm and $N$ is plotted, which is shown in Fig. \ref{fig1}.
\begin{figure}
    \centering
    \includegraphics[height=6.6cm]{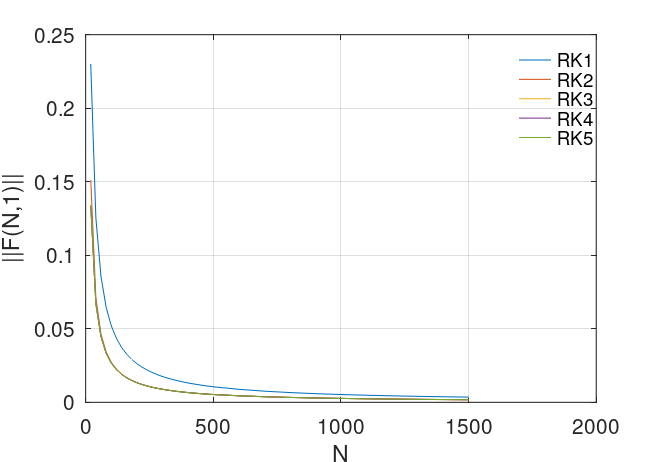}    
\caption{Simulated value of $\Vert F(N,1)\Vert$ versus the number of simulation grids.}  
    \label{fig1}
\end{figure}

From Fig. 1, it is observed that $\Vert F(N,1)\Vert$ approaches 0 as $N$ increases, for all of the five approximation methods applied here. This observation should be viewed in alliance with Theorem 1, which says that $e_N (T)$ converges to $\mathcal{U}_{\rho_0}[T,0](\rho_0-\sigma_0)$ in the long time limit once the accuracy of our chosen approximation method reaches above a certain level. 

However, it is also seen from Fig. 1 that the shapes of a large portion of the curves associated with RK2-RK5 are difficult to distinguish. In other words, the increase of complexity from RK2 to RK5 has \textbf{not} made the relevant curves approach the $N$-axis more rapidly. This observation is viewed in conjunction with Theorem 2, which implies that, as $N$ tends to infinity, the dominating factor of the convergence of $F(N,T)$ to 0 may be \textbf{irrelevant} with numerical approximation methods applied.

Moreover, it is seen that the curves in Fig. 1 resemble inverse proportionality functions. This observation is also viewed together with Theorem 2. On a less rigorous note, by moving the denominator of the l.h.s. of (\ref{s2}) to r.h.s. of the same equation, we arrive at an approximate inverse proportionality relation between $\Vert F(N,T)\Vert$ and $N$. 

Next, we numerically simulate the relation between $\Vert e_N(1)\Vert$ and $\frac{1}{N}$, as shown in Fig. 2.
\begin{figure}
    \centering
    \includegraphics[height=6cm]{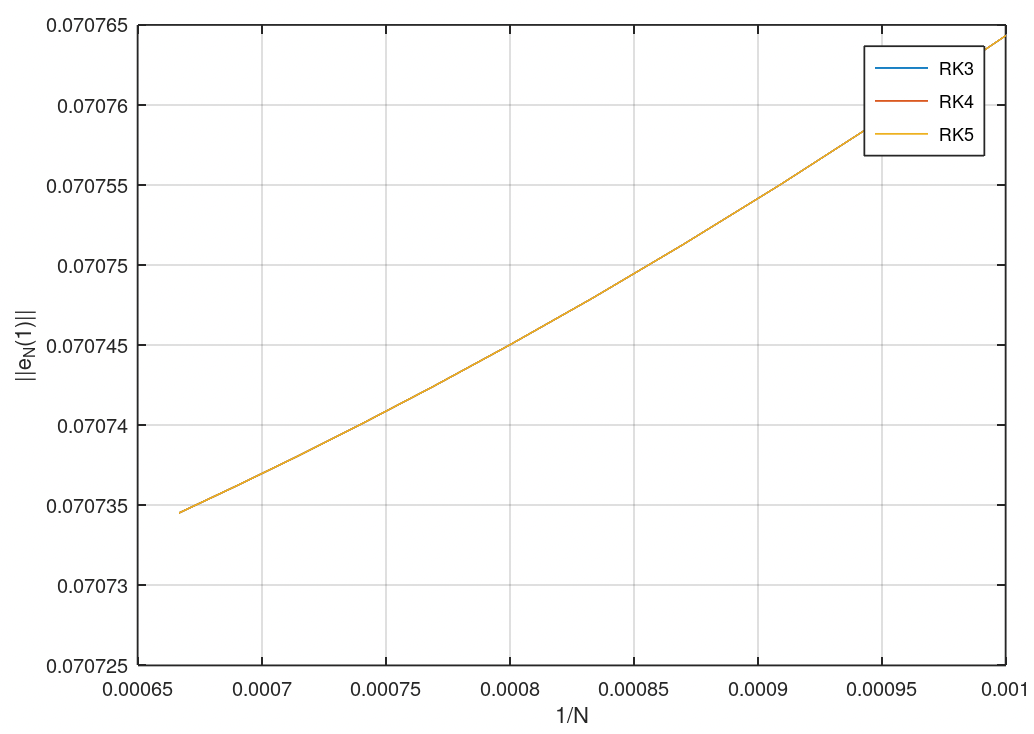}    
\caption{Simulated relation between $\Vert e_N(1)\Vert$ and $\frac{1}{N}$. }  
 \label{fig2}
\end{figure}
Since $T=1$, it is clear that $\frac{1}{N}$ represents the simulation step size. It is observed from Fig. \ref{fig2}
that the curves associated with RK3, RK4 and RK5 are highly overlapped and they all resemble a straight line. We view this observation together with Theorem 3, where (\ref{th3}) shows that, if $2\sum_{n=0}^{N-1}\Vert E(N,n) \Vert \frac{T}{N}$ is a high-order term regarding $\frac{T}{N}$, then the upper bound is roughly linear in the simulation step size for large $N$. 

\section{Conclusion}
We have presented an error analysis for closed-loop designed open-loop control of quantum systems. It is shown that, with a certain level of accuracy for approximating the time-varying Hamiltonian, the error converges to a limit related to initialization error as $N \to +\infty$. With a higher level of accuracy, we have shown that the dominating factor of convergence as $N \to +\infty$ is independent of the approximation method chosen. An upper bound of the error norm is presented, and an example is given to illustrate our results.

\textcolor{blue}{For a comparsion, existing works have covered the closed-loop design phase, e.g., in \cite{KUANG200898,CHEN2017439,CONG20209220,doi:10.1080/00207179.2016.1161830,NOURALLAH2022,KUANG2021109327}, closed-loop control laws have been considered. Acknowledging the fact that theoretically designed closed-loop (possibly nonlinear) systems may not be realized due to laws of quantum dynamics, our work considers the error induced by the application of closed-loop control laws to quantum systems in an open-loop fashion.}



\section{Acknowledgement}                               
This research is partially supported by Hong Kong Research Grant Council (Grants Nos. 15203619 and 15208418), Shenzhen Fundamental Research Fund, China, under Grant No. JCYJ20190813165207290, National Natural Science Foundation of China under Grant No. 62173269, and the CAS AMSS-polyU Joint Laboratory of Applied Mathematics.  

The authors would like to thank Assoc. Prof. Sen Kuang of University of Science and Technology of China for helpful discussions.
\appendix

\section{Proof of Eq. (\ref{appendixshown})}    
Consider the following equation:
\begin{equation}\label{derivative}
\begin{aligned}
\frac{\mathcal{U}_{\rho_0}[T,s](A(s))-\mathcal{U}_{\rho_0}[T,s_0](A(s_0))}{s-s_0}&=\frac{\mathcal{U}_{\rho_0}[T,s](A(s))-\mathcal{U}_{\rho_0}[T,s](A(s_0))}{s-s_0}\\
&+\frac{\mathcal{U}_{\rho_0}[T,s](A(s_0))-\mathcal{U}_{\rho_0}[T,s_0](A(s_0))}{s-s_0}.
\end{aligned}
\end{equation}
On one hand, let us denote
\begin{equation}
    \tilde{A}(s,s_0)\triangleq \frac{A(s)-A(s_0)}{s-s_0}.
\end{equation}
It is shown that
\begin{equation}\label{triangle}
    \begin{aligned}
    &\Vert\mathcal{U}_{\rho_0}[T,s](\tilde{A}(s,s_0))-\mathcal{U}_{\rho_0}[T,s_0](\dot{A}(s_0))\Vert=\Vert\mathcal{U}_{\rho_0}[T,s](\tilde{A}(s,s_0))-\mathcal{U}_{\rho_0}[T,s_0](\tilde{A}(s,s_0))\\ 
    &+\mathcal{U}_{\rho_0}[T,s_0](\tilde{A}(s,s_0))-\mathcal{U}_{\rho_0}[T,s_0](\dot{A}(s_0))\Vert\\ 
    &\leq \Vert\mathcal{U}_{\rho_0}[T,s]-\mathcal{U}_{\rho_0}[T,s_0]\Vert\cdot\Vert\tilde{A}(s,s_0)\Vert +\Vert\mathcal{U}_{\rho_0}[T,s_0]\Vert\cdot\Vert\tilde{A}(s,s_0)-\dot{A}(s_0)\Vert.
    \end{aligned}
\end{equation}
Therefore, it is clear from (\ref{triangle}) that
\begin{equation}\label{appendixfirst}
\begin{aligned}
\lim_{s\to s_0}\frac{\mathcal{U}_{\rho_0}[T,s](A(s))-\mathcal{U}_{\rho_0}[T,s](A(s_0))}{s-s_0}&=\lim_{s\to s_0}\mathcal{U}_{\rho_0}[T,s](\tilde{A}(s,s_0))\\
&=\mathcal{U}_{\rho_0}[T,s_0](\dot{A}(s_0)).
\end{aligned}
\end{equation}

On the other hand, 
\begin{equation}
\begin{aligned}
\frac{\mathcal{U}_{\rho_0}[T,s](A(s_0))-\mathcal{U}_{\rho_0}[T,s_0](A(s_0))}{s-s_0}&=\frac{\mathcal{U}_{\rho_0}[T,s](A(s_0))-\mathcal{U}_{\rho_0}[T,s]\big(\mathcal{U}_{\rho_0}[s,s_0](A(s_0))\big)}{s-s_0}\\
&=-\frac{\mathcal{U}_{\rho_0}[T,s]\big((\mathcal{U}_{\rho_0}[s,s_0]-\mathcal{I})(A(s_0))\big)}{s-s_0}.
\end{aligned}
\end{equation}
It thus follows that
\begin{equation}\label{appendixsecond}
\begin{aligned}
\lim_{s\to s_0}\frac{\mathcal{U}_{\rho_0}[T,s](A(s_0))-\mathcal{U}_{\rho_0}[T,s_0](A(s_0))}{s-s_0}&=-\lim_{s\to s_0}\frac{\mathcal{U}_{\rho_0}[T,s]\big((\mathcal{U}_{\rho_0}[s,s_0]-\mathcal{I})(A(s_0))\big)}{s-s_0}\\
&=-\lim_{s\to s_0}\mathcal{U}_{\rho_0}[T,s]\Bigg(\frac{(\mathcal{U}_{\rho_0}[s,s_0]-\mathcal{I})(A(s_0))}{s-s_0}\Bigg)\\
&=\mathcal{U}_{\rho_0}[T,s_0]\big(\text{i}[H_{\rho_0}(s_0),A(s_0)]\big).
\end{aligned}
\end{equation}
Finally, with the aid of (\ref{appendixfirst}) and (\ref{appendixsecond}), we have
\begin{equation}
\begin{aligned}
&\lim_{s\to s_0}\frac{\mathcal{U}_{\rho_0}[T,s](A(s))-\mathcal{U}_{\rho_0}[T,s_0](A(s_0))}{s-s_0}\\
&=\mathcal{U}_{\rho_0}[T,s_0](\dot{A}(s_0))+\mathcal{U}_{\rho_0}[T,s_0]\big(\text{i}[H_{\rho_0}(s_0),A(s_0)]\big),
\end{aligned}
\end{equation}
which completes the proof.

\section{Proof of Eq. (\ref{thm3_6})}
Let us denote
\begin{equation}
    -\text{i}[H_{\rho_0}(s),\rho(s;\rho_0,0)]\triangleq \mathcal{L}_H (s),
\end{equation}
for $s>0$. It is clear that, for $s>0$,
\begin{equation}
    \begin{aligned}
    &\big( \mathcal{L}_H (s) \big)_{ii}=\dot{\rho}_{ii}(s),\quad 1\leq i \leq N_S;\\
    &\big( \mathcal{L}_H (s) \big)_{N_S N_S}=-\sum_{n=1}^{N_S-1}\dot{\rho}_{nn}(s);\\
    &\big( \mathcal{L}_H (s) \big)_{jk}=\dot{\rho}_{jk}^R(s)+\text{i}\dot{\rho}_{jk}^I(s),\quad 1\leq j<k \leq N_S;\\
    &\big( \mathcal{L}_H (s) \big)_{jk}=\dot{\rho}_{kj}^R(s)-\text{i}\dot{\rho}_{kj}^I(s),\quad 1\leq k<j \leq N_S.
    \end{aligned}
\end{equation}
Next, we compute that (omitting variable ``$s$" hereafter):
\begin{equation}
\begin{aligned}
\text{tr}(D_{u_k} L_H)&=\sum_{j=1}^{N_S} (D_{u_k} L_H)_{jj}\\ 
                      &=\sum_{j=1}^{N_S} \sum_{k=1}^{N_S} (D_{u_k})_{jk} (L_H)_{kj}\\
                      &=\sum_{i=1}^{N_S-1}(D_{u_k})_{ii}(L_H)_{ii}+(D_{u_k})_{N_S N_S}(L_H)_{N_S N_S}\\
                      &+\sum_{1\leq j<k \leq N_S}(D_{u_k})_{jk} (L_H)_{kj}+\sum_{1\leq k<j \leq N_S}(D_{u_k})_{jk} (L_H)_{kj}.
\end{aligned}
\end{equation}

Firstly, we have
\begin{equation}\label{B.4}
\begin{aligned}
\sum_{i=1}^{N_S-1}(D_{u_k})_{ii}(L_H)_{ii}+(D_{u_k})_{N_S N_S}(L_H)_{N_S N_S}&=\sum_{i=1}^{N_S-1}\hat{u}_{ii}\dot{\rho}_{ii}+0 \cdot (L_H)_{N_S N_S}\\
&=\sum_{i=1}^{N_S-1}\hat{u}_{ii}\dot{\rho}_{ii}.
\end{aligned}
\end{equation}

Secondly, we derive that
\begin{equation}\label{B.5}
\begin{aligned}
\sum_{1\leq j<k \leq N_S}(D_{u_k})_{jk} (L_H)_{kj}&=\sum_{1\leq j<k \leq N_S}\Bigg(\frac{\hat{u}_{jk}^R +\text{i}\hat{u}_{jk}^I}{2} \Bigg)(\dot{\rho}_{jk}^R-\text{i}\dot{\rho}_{jk}^I)\\
&=\sum_{1\leq j<k \leq N_S}\frac{1}{2}(\hat{u}_{jk}^R \dot{\rho}_{jk}^R-\text{i}\hat{u}_{jk}^R \dot{\rho}_{jk}^I+\text{i}\hat{u}_{jk}^I \dot{\rho}_{jk}^R +\hat{u}_{jk}^I \dot{\rho}_{jk}^I).
\end{aligned}
\end{equation}

Thirdly, we yield that 
\begin{equation}\label{B.6}
\begin{aligned}
\sum_{1\leq k<j \leq N_S}(D_{u_k})_{jk} (L_H)_{kj}&=\sum_{1\leq k<j \leq N_S}\Bigg(\frac{\hat{u}_{kj}^R -\text{i}\hat{u}_{kj}^I}{2} \Bigg)(\dot{\rho}_{kj}^R+\text{i}\dot{\rho}_{kj}^I)\\
&=\sum_{1\leq k<j \leq N_S}\frac{1}{2}(\hat{u}_{kj}^R \dot{\rho}_{kj}^R+\text{i}\hat{u}_{kj}^R \dot{\rho}_{kj}^I-\text{i}\hat{u}_{kj}^I \dot{\rho}_{kj}^R +\hat{u}_{kj}^I \dot{\rho}_{kj}^I).
\end{aligned}
\end{equation}

Finally, we combine (\ref{B.4}), (\ref{B.5}), and (\ref{B.6}), yielding that
\begin{equation}\label{B.7}
    \text{tr}(D_{u_k} L_H)=\sum_{i=1}^{N_S-1}\hat{u}_{ii}\dot{\rho}_{ii}+\sum_{1\leq j<k \leq N_S}\big(\hat{u}_{jk}^R \dot{\rho}_{jk}^R+\hat{u}_{jk}^I \dot{\rho}_{jk}^I\big).
\end{equation}
The r.h.s. of (\ref{B.7}) agrees with the r.h.s. of (\ref{thm3_4}), which indicates that (\ref{thm3_6}) holds. The proof is completed.

 \bibliographystyle{elsarticle-num} 
 \bibliography{cas-refs}





\end{document}